\documentclass[a4paper,english,numberwithinsect]{eurocg22}

\usepackage{microtype} % if unwanted, comment out
\usepackage{comment}
\title{Finding a Battleship of Uncertain Shape\footnote{This work was initiated at the 2nd Austrian Computational Geometry Reunion Workshop in Strobl, June 2021. E.-M. H. supported by the Austrian Science Foundation FWF, project F55-02. D.P. partially supported by FWF within the collaborative DACH project \emph{Arrangements and Drawings} as FWF project \mbox{I 3340-N35}. M.L. partially supported by the Dutch Science Foundation (NWO) under grant number 614.001.504}}
\titlerunning{Finding a Battleship}

\author[1]{Eva-Maria Hainzl}
\author[2]{Maarten L\"offler}
\author[3]{Daniel Perz}
\author[4]{Josef Tkadlec}
\author[5]{Markus Wallinger}
\affil[1]{Institute of Discrete Mathematics and Geometry, TU Wien\\
  \texttt{eva-maria.hainzl@tuwien.ac.at}}
\affil[2]{Department of Computing and Information
        Sciences, Utrecht University\\
  \texttt{m.loffler@uu.nl}}
\affil[3]{Institute of Software Technology, TU Graz\\
  \texttt{daperz@ist.tugraz.at}}
\affil[4]{Department of Mathematics, Harvard University\\
  \texttt{tkadlec@math.harvard.edu}}
\affil[5]{Algorithms and Complexity Group, TU Wien\\
  \texttt{mwallinger@ac.tuwien.ac.at}}
\authorrunning{E.-M. Hainzl, M. Löffler, D. Perz, J. Tkadlec, M. Wallinger}
\ArticleNo{48}

\usepackage{hyperref}
\usepackage[capitalise]{cleveref}
\usepackage{color}
\usepackage{fdsymbol}
\newcommand{\F}{\mathcal{F}}

\newcommand{\den}{\pi}
\newcommand{\spn}{\operatorname{sp}}

\newcommand{\Z}{\mathbb{Z}}
\newcommand{\N}{\mathbb{N}}

\newcommand{\formatshapecell}[1]{
\ifthenelse{\equal{#1}{.}}{{\includegraphics[scale=.4,page=4]{figures/symbols}}}{
\ifthenelse{\equal{#1}{o}}{{\includegraphics[scale=.4,page=2]{figures/symbols}}}{
?}}
}
\newcommand{\formatpatterncell}[1]{
\ifthenelse{\equal{#1}{.}}{{\includegraphics[scale=.4,page=4]{figures/symbols}}}{
\ifthenelse{\equal{#1}{x}}{{\includegraphics[scale=.4,page=3]{figures/symbols}}}{
?}}
}

\makeatletter

\newcommand{\shape}[1]{\ensuremath{
  [\@tfor\next:=#1\do{\formatshapecell{\next}}]}%
}
\newcommand{\pattern}[1]{\ensuremath{
  (\@tfor\next:=#1\do{\formatpatterncell{\next}})}%
}
\makeatother
\newcommand{\reverse}[1]{\overline{#1}}

\newtheorem*{theorem*}{Theorem}

\begin{document}

\maketitle

\begin{abstract}
Motivated by a game of \emph{Battleship}, we consider the problem of efficiently hitting a ship of an uncertain shape within a large playing board.
Formally, we fix a dimension $d\in\{1,2\}$.
A ship is a subset of $\Z^d$.
Given a family $\F$ of ships, we say that an infinite subset $X\subset\Z^d$ of the cells \emph{pierces} $\F$, if it intersects each translate of each ship in $\F$ (by a vector in $\Z^d$).
In this work, we study the lowest possible (asymptotic) density $\den(\F)$ of such a piercing subset.
To our knowledge, this problem has previously been studied only in the special case $|\F|=1$ (a single ship).
As our main contribution, we present a formula for $\den(\F)$ when $\F$ consists of 2 ships of size 2 each,
and we identify the toughest families in several other cases.
We also implement an algorithm for finding $\den(\F)$ in 1D.

\end{abstract}

 \begin {figure}
   \mbox{}\hfill
   \includegraphics [page=1] {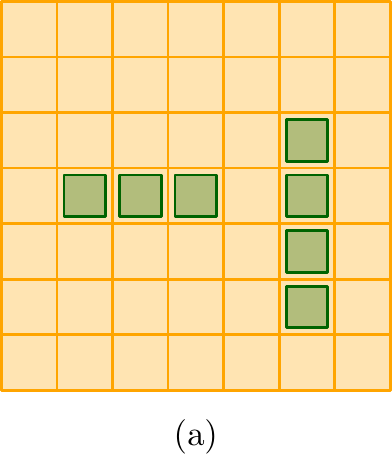} \hfill
   \includegraphics [page=2] {figures/example-2d.pdf} \hfill
   \includegraphics [page=3] {figures/example-2d.pdf} \hfill
   \mbox{}
   \caption 
   {
     (a) A game with two ships: a $3 \times 1$ rectangle and a $1 \times 4$ rectangle. Note that we {\bf do not} allow individual ships to be rotated.
     (b) A shooting pattern that is certain to hit both ships, no matter where they are.
     (c) A sparser (in fact, optimal) shooting pattern.
   }
   \label {fig:example-2d}
 \end {figure}

\section{Introduction}
In a game \emph {Battleship}, two players first secretly place a family $\F$ of \emph {ships} (often rectangles) on a \emph{domain} (an integer grid),
and then they aim to locate the opponent's ships by querying individual grid cells.
Inspired by the game, we consider the problem of finding a sparse \emph {shooting pattern}:
a subset of the grid cells that is guaranteed to hit at least one cell of each ship,
no matter how the ships are translated within the domain.
See Figure~\ref {fig:example-2d} for an example and Section~\ref {sec:prelims} for a formal problem statement.

This problem is surprisingly intricate.
In this note, we make two simplifying assumptions.

\subparagraph {Infinite domains.}
First,
in order to avoid boundary effects, we assume
the domain is an infinite grid $\Z^d$.
Since any shooting pattern on an infinite domain is also infinite,
we measure its quality using the (asymptotic) \emph {density},
refer to Figure~\ref {fig:infinite-2d}.
Note that the problem is subtle;
for instance, for two $L$-shaped triominoes as ships, the lowest possible density of a shooting pattern depends on the relative orientation of the ships
(see Figure~\ref {fig:double-2d}). 

 \begin {figure}
   \mbox{}\hfill
   \includegraphics [page=1] {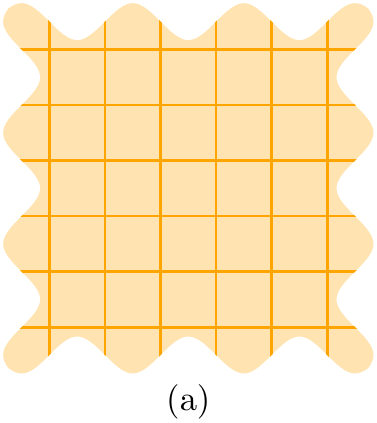} \hfill
   \includegraphics [page=2] {figures/infinite-2d.pdf} \hfill
   \includegraphics [page=3] {figures/infinite-2d.pdf} \hfill
   \mbox{}
   \caption 
   {
     (a) An infinite playing board.
     (b) A single $L$-shaped ship. It may be translated, but not rotated.
     (c) An optimal shooting pattern with density $\frac13$ hitting every possible translation.
        }
   \label {fig:infinite-2d}
 \end {figure}

\begin{lemma}\label{lem:2l}
Let $\F_{180}$ and $\F_{90}$ be families of two L-shaped triominoes from~\cref{fig:double-2d}. Then 
$\den(\F_{180})=\frac13$ and $\den(\F_{90})= \frac12$.
\end{lemma}
\begin{proof}
The shooting patterns shown in \cref{fig:double-2d} imply
 $\den(\F_{180})\le \frac13$ and $\den(\F_{90})\le\frac12$.
For $\F_{180}$ the matching lower bound is trivial, since $\den(\F_{180})\ge\den(\F)\ge\frac13$, where $\F$ consists of a single L-shaped triomino.
It remains to prove $\den(\F_{90})\ge \frac12$.
Split the plane into infinite vertical slabs of width~2.
We argue for each slab separately.
Fix a row. If neither cell is shot, then in the row just above it, both cells must be shot.
Hence the overall density is at least~$\frac12$.
\end{proof}

 \begin {figure}
   \mbox{}\hfill
   \includegraphics [page=1] {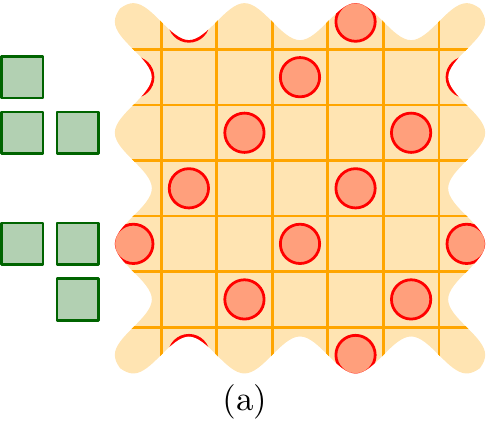} \hfill
   \includegraphics [page=2] {figures/double-2d.pdf} \hfill
   \mbox{}
   \caption 
   {
     Two sets of two L-shaped ships of size $3$.
     (a) Two ships, rotated $180^\circ$.
     The same optimal shooting pattern with density $\frac13$ as for a single ship still works.
     (b) Two ships, rotated $90^\circ$. The optimal shooting pattern has density $\frac12$ (see Lemma~\ref{lem:2l}).
   }
   \label {fig:double-2d}
 \end {figure}

\subparagraph {One dimension.}
Second, for the remainder of this note we focus on the case $d=1$, which is far from trivial
if we consider not only \emph {connected} ships, but arbitrary finite subsets of integers (see Section~\ref {sec:prelims}).
In~\cref{app:d=2}, we describe how our results extend to higher dimensions.
Thus, our problem is:
Given a family $\F=\{S_1,\dots,S_n\}$ of $n$ ships in $\Z$ (or in $\Z^d$),
find the minimum density $\den(\F)$ of a shooting pattern that hits each translate of each ship $S_i\in\F$.

% related work
\subsection{Related work}
The game \emph{Battleship} spawned research along several fronts~\cite{crombez_et_al,fiat1989find}.
Here we review the 1D~case.

We say that a family $\F$ of ships in $\Z$ is of type $(k_1,\dots,k_n)$ if it consists of $n$ ships with sizes $k_1\le\dots\le k_n$, respectively.
Previous work has studied families consisting of a single (not necessarily connected) ship, that is, families of type $(k)$ for some $k\in\N$.
It is easy to see that
$\den(\F)=1/k$ for any $(k)$-family $\F$ with $k\in\{1,2\}$.
In 2008, Schmidt and Tuller conjectured a formula for $\den(\F)$ for any $(3)$-family $\F$~\cite{schmidt2008covering},
but as of now its validity is still open.

Given this difficulty, other works studied the toughest instances of a given type.
Formally, given a type $t$, let
$M_t = \sup_{\F\textrm{ has type $t$}} \{\den(\F)\}
$
be the smallest density that suffices to hit any family of type $t$.
Already in 1967, Newman~\cite{newman1967complements} showed that $M_{(3)}=\frac25$
(one toughest instance is the ship in \cref{fig:example-1d}(a))
and that $M_k=\Theta(\log k /k)$ as $k\to\infty$.
Recently, it was shown that $M_{(4)}=\frac13$~\cite{talbot2019polychromatic}.
Also, given a ship $S$,
the density $\den(\{S\})$
can be found using a ``sliding window'' algorithm~\cite{bollobas2011covering}.
The algorithm can be used to establish lower bounds such as $M_{(5)}\ge \frac3{11}$.

To our knowledge, the problem for multiple ships
has not been studied; however, we point out a similar work in the continuous setting for rectangles in two dimensions~\cite{dumitrescu2021piercing}.

% our contribution
\subsection{Our contribution}
We propose to study the problem for multiple ships or, equivalently, for a single ship of uncertain shape.
(Note that by considering suitable families, we can model mirrored or reflected ships on top of translated ships, cf.~\cref{fig:double-2d}.)
Apart from Lemma~\ref{lem:2l} above, we present results in 1D
(see~\cref{app:d=2} for extensions to 2D).
First, we note that the sliding window algorithm of~\cite{bollobas2011covering} can be adapted to families of multiple ships in a straightforward way, see~\cref{thm:algorithm}. We implement the algorithm and use it to obtain lower bounds such as $M_{(2,3)}\ge \frac35$ (due to e.g. $\F=\{[0,1],[0,2,4]\}$).
% Running the algorithms, we obtain lower bounds such as $M_{(2,3)}\ge \frac35$
% (due to e.g. $\F=\{[0,1],[0,2,4]\}$) % and $\F=\{[0,2],[0,3,4]\}$
% and $M_{(3,3)}\ge \frac12$
% (due to e.g. $\F=\{[0,1,2],[0,2,3]\}$), see~\cref{fig:table}.
As our main contribution, we present three results for families of ships of small size $k$ ($k$-ships).
First, for any family $\F$ of two 2-ships, we find an explicit formula for $\den(\F)$, see~\cref{thm:formula22}.
Second, we determine $M_{(2,\dots,2)}$, that is, we identify the toughest instances for families consisting of any number of 2-ships, see~\cref{thm:n2}.
Third, we determine the toughest instances for families consisting of any 3-ship together with its reflection, see~\cref{thm:32}. 
Finally, in \cref{app:general-ub}
we present general bounds for the density of the toughest instances of $n$ ships of~size~$k$~each.

% preliminaries
\subsection{Preliminaries} \label {sec:prelims}
%\begin{definition}
A ship of size $k$ (a \emph{$k$-ship}) is a $k$-tuple
$[a_1, a_2, \dots a_k]$ with $a_i \in \mathbb{Z}, i\geq 1$.
A~\emph{span} of a ship $S$ is
$\spn(S)=a_k-a_1+1$.
See~\cref{fig:example-1d} for an illustration. 
A finite family of ships $\F = \{S_1,S_2,\dots S_n\}$ has a span $\spn(\F) = \max_{S \in \F} \spn(S)$.
A \emph{shooting pattern} is a $01$-sequence $X = (x_i)_{i \in \mathbb{Z}}$.
We say that a shooting pattern \emph{hits} a $k$-ship
$S=[a_1, a_2,\dots a_k]$
(or that $X$ is a shooting pattern for $S$) if
    $\sum_{i=1}^{k} x_{n+a_i} \geq 1,\; \forall n\in \mathbb{Z}$.
The density of a shooting pattern $X$ is defined as 
$\pi(X) = \lim_{N\rightarrow \infty} \frac{\sum_{|i|\leq N}x_i}{2N+1}.$
The density of a ship $S$ and of a family $\F$ of ships is then defined as
\[\pi(S) = \inf_{X\;\text{hits}\;S} \pi(X) \quad \text{and}\quad
\pi(\F) = \inf_{X\;\text{hits each}\;S\in\F} \pi(X).
\]

 \begin {figure}
   \mbox{}\hfill
   \includegraphics [page=1] {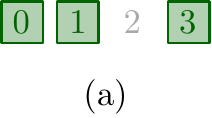} \hfill
   \includegraphics [page=2] {figures/example-1d.pdf} \hfill
   \includegraphics [page=3] {figures/example-1d.pdf} \hfill
   \mbox{}
   \caption 
   {
     (a) A~single disconnected 1-dimensional ship $S=[0,1,3]$ of size $k=3$ and span $\spn(S)=4$.
     (b-c) Possible periodic shooting patterns for $\{S\}$, with densities $\frac12$ and $\frac25$, respectively.
   }
   \label {fig:example-1d}
 \end {figure}
 
Further, we define $m^n_k=\inf\{ \den(\{S_1,\dots,S_n\}) \mid |S_i|=k \text{ for } i=1,\dots,n \}$ and $M^n_k=\sup\{ \den(\{S_1,\dots,S_n\}) \mid |S_i|=k \text{ for } i=1,\dots,n \}$.
That is, $m^n_k$ is the required density for the simplest instances,
whereas $M^n_k$ is the required density for the toughest instances, among families that consist of $n$ ships of size $k$ each.
Regarding $m_k^n$, it is straightforward to prove $m_k^n = \frac{1}{k}$ (even when no two ships in the family are translates of each other).
Regarding $M^n_k$, in~\cref{app:general-ub} we show $1-\frac{e}{\sqrt[k-1]{n}} \le M_k^n\le \min\left\{\frac{n}{n+1},\frac{1+\log(kn)}{k}\right\}$ when $k\ge 2$.
%The lower bound can be proven by choosing a block of length $m$ such that $n\geq \binom{m}{k}$ and taking every $k$-ships which fits inside the block.
%The  proof of the upper bound is based on a paper by Alon~\cite{alon-1990} who computed an upper bound on minimal transversals of hypergraphs. 
%\cref{app:general-ub}.
%
We also list two auxiliary results (with proofs in \cref{app:algos}).
\begin{theorem}[Sliding window algorithm]\label{thm:algorithm}
Given a family $\F$ with span $s=\spn(\F)$, the density $\den(\F)$ can be computed in time polynomial in $2^s$.
\end{theorem}

Finally, for an integer $d$ and a ship $S = [a_1, \ldots, a_k]$ let $dS = [da_1, \ldots, da_k]$, and likewise for a family $\F = \{S_1, \ldots, S_n\}$ let $d\F = [dS_1, \ldots, dS_n]$.
\begin{lemma} \label{lem:scaling}
  Let $d$ be a positive integer and $\F$ any family. Then $\pi(\F) = \pi(d\F)$.
\end{lemma}

\section{Families of 2-ships and 3-ships} \label {sec:small}
Here we study families $\F$ that consist of $n$ ships of small size $k\le 3$ each.
Note that when $k=1$, we obviously have $\den(\F)=1$ (for all $n\ge 1$).
Also, for a single 2-ship $S$ it is straightforward to show that $\den(\{S\})=1/2$.
Our first non-trivial result is an explicit formula for $\den(\F)$
when $\F$ consists of two 2-ships.

\begin{theorem}[Formula for two 2-ships]\label{thm:formula22}
Let $\F=\{[0,da],[0,db]\}$, $a$, $b$ coprime and $d\ge 1$.
\[
\den(\F) = 
\begin{cases}
    1/2                    & \text{if both } a \text{ and } b \text{ are odd,}\\
    \frac{a+b+1}{2(a+b)}   & \text{otherwise.}
\end{cases}
\]
\end{theorem}
\begin{proof}
Let $\F'=\{[0,a],[0,b]\}$.
By Lemma~\ref {lem:scaling}, it suffices to determine $\den(\F')$. Clearly, $\den(\F')\ge \den(\{[0,a]\}=1/2$.
When both $a$ and $b$ are odd, a shooting pattern $X$ defined by ``$x_i=1$ if and only if $i$ is even'' provides a matching construction.
From now on, assume that precisely one of $a$, $b$ is odd (that is, $a+b$ is odd).

Split $\Z$ into blocks of $a+b$ consecutive integers.
Let $S=\{0,\dots,a+b-1\}$ be one such block and let $X'$ be a shooting pattern for $\F'$ on $S$ (instead of on $\Z$).
We will argue that $S$ needs to be hit at least $(a+b+1)/2$ times (that is, $\sum_{i\in S} x'_i\ge(a+b+1)/2$).
Consider a graph $G=(S,E)$ with nodes $S$ and directed edges $E=\{(u,v) \mid v-u\in\{a,-b\}\}$.
This graph records the constraints on the shooting pattern:
For every edge $(u,v)\in E$, we must have $x_u+x_v\ge 1$.
Since $|S|=a+b$, each node in $G$ has indegree 1 and outdegree 1.
Moreover, since $a$ and $b$ are coprime, the graph $G$ is connected.
Hence it is a directed cycle on an odd number $a+b$ of nodes.
Its minimum vertex cover has size $(a+b+1)/2$, thus $\den(\F')\ge (a+b+1)/2$.

To prove that this bound is tight, consider any vertex cover $C\subseteq S$ of $G$ of the minimum size $(a+b+1)/2$.
Then the $(a+b)$-periodic shooting pattern $X^C$ defined by ``$x^C_i=1$ if and only if $i\pmod{(a+b)}\in C$''
hits $\F'$ on $\Z$:
Indeed, consider any translate $(n,n+a)$ of the ship $[0,a]$. Suppose $n\equiv r\pmod{(a+b)}$ for some $0\le r< a+b$. If $r<b$ then $n$ and $n+a$ both belong to the same block, thus $x_n+x_{n+a}\ge 1$, since $C$ is a vertex cover. On the other hand, if $r\ge b$ then by the $(a+b)$-periodicity of the shooting pattern we have $x_{n+a}=x_{n-b}$. Since $r\ge b$, both $n-b$ and $n$ belong to the same block, so we conclude as before. For translates of the ship $[0,b]$ we argue analogously.
\end{proof}
As a corollary, we have $M^2_2=2/3$, as witnessed by families $\{[0,d],[0,2d]\}$ for any $d\ge 1$.

Next, we study the toughest instances in two other cases,
namely for any number of $2$-ships, and for a 3-ship together with its reflection.
\begin{theorem}[Toughest families of 2-ships]\label{thm:n2}
 For any $n\ge 1$ we have $M^n_2=n/(n+1)$.
\end{theorem}
\begin{proof}
For $n=1$ the claim is trivial. For $n=2$ it follows from~\cref{thm:formula22}.
Assume $n\ge 3$.

First, note that for a family $\F_n=\{[0,1],\dots,[0,n]\}$ we have $\den(\F_n)=n/(n+1)$:
Indeed, split $\Z$ into blocks of $n+1$ consecutive integers.
Then any shooting pattern $X_n$ for $\F_n$ may miss at most 1 number from each block.
On the other hand, the pattern $X_n$ defined by ``$x_i=0$ if and only if $i$ is a multiple of $n+1$'' hits $\F_n$.

To prove the upper bound, consider any $n$ positive integers
$a_1<\dots<a_n$,
and the corresponding family $\F=\{[0,a_1],\dots,[0,a_n]\}$.
We construct a shooting pattern $X$ for $\F$ with density at most $n/(n+1)$.
We proceed in steps.
Initially, we set $x_t=1$ for all $t$ with $|t|\le a_n$.
Then, we process integers $t> a_n$ in increasing order.
Whenever $x_t$ is not yet set, we set $x_t=0$ and $x_{t+a_i}=1$ for each $i=1,\dots,n$.
(Note that some of $x_{t+a_i}$ might have already been set to 1, due to some $t'<t$.)
By construction, $X$ hits all translates of $\F$ within the interval $[a_n+1,\infty)$.
Moreover, since for every $x_t$ set to $0$ there are at most $n$ values newly set to 1,
in the limit $t\to\infty$ we obtain $\den(X[a_n+1,\infty))\le n/(n+1)$.
Similarly, we process integers $t<-a_n$ in decreasing order and get $\den(X(-\infty,-a_n-1])\le n/(n+1)$.
Together with the finite initial segment
$X[-a_n,a_n]$
%which satisfies $\den(X[-a_n,a_n])=1$,
this gives $\den(X)=\den(X(-\infty,\infty))\le n/(n+1)$.% as desired.
\end{proof}

Given a ship $S=[a_1,\dots,a_k]$, let $\reverse{S}=[-a_k,\dots,-a_1]$ be its reflection.

\begin{theorem}[Toughest 3-ship with its reflection]\label{thm:32} 
Let $S$ be a 3-ship and let $\F=\{S,\reverse{S}\}$. Then
$\den(\F)\le \frac25$, with equality if and only if $S\in\{[0,2d,3d],[0,3d,4d]\}$ (or their reflections) for some $d\ge 1$.
\end{theorem}

 \begin {figure}
   \mbox{}\hfill
   \includegraphics [page=1] {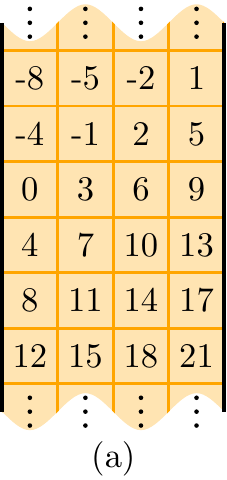} \hfill
   \includegraphics [page=2] {figures/vertical-slab.pdf} \hfill
   \includegraphics [page=3] {figures/vertical-slab.pdf} \hfill
   \includegraphics [page=4] {figures/vertical-slab.pdf} \hfill
   \includegraphics [page=5] {figures/vertical-slab.pdf} \hfill
   \mbox{}
   \caption 
   {
     Illustration of the solution for $S=[0,4,7]$.
     (a) A layout of the integers into an infinite vertical slab of width $4$.
     (b-c) Every translation of $S$ corresponds to a triple of cells that form an $L$-shape.
     (d) If the $L$-shape falls on the edge of the slab, the pieces ``wrap around'' but are shifted vertically.
     (e) A valid shooting pattern.
   }
   \label {fig:vertical-slab}
 \end {figure}
 
\begin{proof} 
  First, we argue that for a single 3-ship $S$, the toughest instance has a density of $\frac25$; that is, $M^1_3=\frac25$.
  This fact was first proven by Newman~\cite{newman1967complements}.
  Here, we present a geometric proof, which we then extend to the case of two symmetric 3-ships.
  
  %We claim that $\den(\{S\})\le 2/5$ with equality for precisely $[0,2,3]$ or $[0,3,4]$ (or their reverses).
  Consider a 3-ship $S = [0, a, a+b]$ for positive integers $a$ and $b$ with $\operatorname{GCD}(a,b)=1$ and $a\ge b$. We arrange the integers into a 2-dimensional grid $\{0, \ldots, a-1\} \times \Z$ by the bijection $(i,j) \mapsto ib + ja$.
  Refer to Figure~\ref {fig:vertical-slab}(a).
  Most translations of $S$ now correspond to an L-triomino with the same orientation; therefore, we can hit all of them with a shooting pattern with density $\frac13$ using the same solution as in Figure~\ref {fig:infinite-2d}.
  However, this misses exactly the translations by an amount that is congruent to $-b \mod a$; those translations correspond to a triomino that ``wraps around'' (Figure~\ref {fig:vertical-slab}(d)). To hit these it is sufficient to increase the density of the first column to $\frac23$. This gives $\den(\{S\})\le (a+1)/(3a)$, which is strictly less than $2/5$ for $a\ge 6$. %For $[0,2,3]$ and $[0,3,4]$ the best density is indeed $2/5$ because every gap must be preceded by a few marks (small casework).
  For the remaining $10$ cases with $b<a\le 5$, we find an optimal solution by~\cref{thm:algorithm}. The toughest instances, yielding $\den(S)=\frac25$, turn out to be $S\in\{[0,2,3],[0,3,4]\}$ as claimed. % finds a solution with density $<2/5$.
  
  Now, fix $S$ and consider a family $\F = \{S, \reverse{S}\}$.
  We claim that $\pi(\F) \le 2/5$ as well.
  Indeed, as in Lemma~\ref {lem:2l}, when $a \ge 6$, the solution described above hits not only every translate of $S$ but also every translate of $\reverse{S}$.
  For the 10 cases with $b<a \le 5$, using~\cref{thm:algorithm} we again verify that all such families $\F=\{S,\reverse{S}\}$ satisfy $\den(\F)\le \frac25$, with equality for $S\in\{[0,2,3], [0,3,4]\}$.
%   for the families
%   $\{[0, 2, 3], [0, 1, 3]\}$ and
%   $\{[0, 3, 4], [0, 1, 4]\}$
%   (and $\{[0, 4, 6], [0, 2, 6]\}$, which is the same as $\{[0, 2, 3], [0, 1, 3]\}$ \maarten {by Lemma~\ref {lem:scaling}}).
\end{proof}

We note that Theorems \ref{thm:formula22}, \ref{thm:n2} and \ref{thm:32} can be generalized to higher dimensions. Notes on these extensions can be found in~\cref{app:d=2}.

\begin{figure}[ht]
  \centering
  \includegraphics[width=0.8\linewidth]{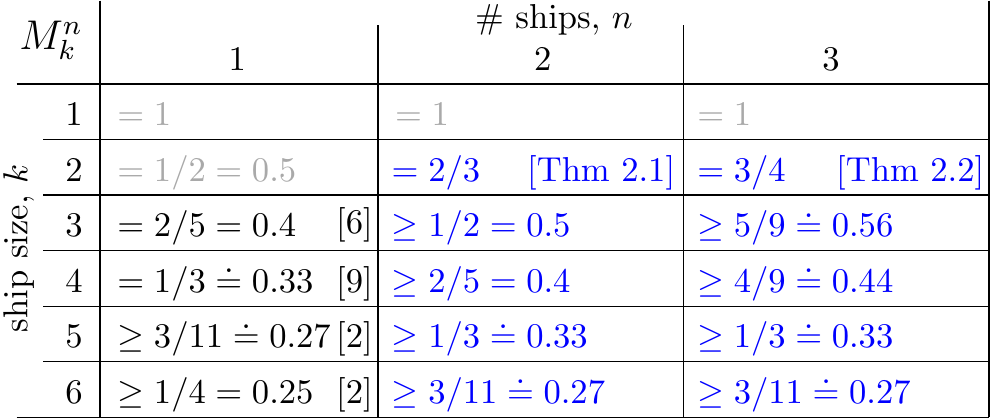}
  \caption{Bounds on the density $M^n_k$ of the toughest instances among the families with $n$ ships of size $k$ each, found by running the algorithm in~\cref{thm:algorithm} for all families with span up to $11-n$.}
  \label{fig:table}
\end{figure}

\section{Conclusions}
We introduced the problem of locating a battleship of an uncertain shape.
Given the difficulty of the problem in general, we focused on the simplest possible setting, namely ships of size 2 or 3 in 1D (see~\cref{app:d=2}
for extensions to 2D).
We also implemented an algorithm for computing $\den(\F)$ in 1D and used it to compute lower bounds on the minimum density $M^n_k$ required for the toughest families of $n\le 3$ ships of size $k\le 6$ each, see~\cref{fig:table}.
Many open problems arise, e.g.:
\begin{enumerate}
    \item Which values in~\cref{fig:table} are tight? For instance, is it true that $M^2_3=1/2$?
    \item What are the asymptotics of $1-M^n_k$ for fixed small $k\ge3$? %\pepa{state $k=2$ from gen thm?}
    \item In 2D, is there an algorithm for computing $\den(\F)$?
\end{enumerate}
%\item In 1D, the situation is simplified by the fact that there exists an optimal shooting pattern which is periodic.
%Does some analogous statement hold in 2D?

%\bibliography{bibliography}

\appendix

\section{Proof of~\cref{thm:algorithm} and Lemma~\ref{lem:scaling}}\label{app:algos}

\textbf{Proof of~\cref{thm:algorithm}.}
This is a direct extension of~\cite[Remark 5.4]{bollobas2011covering}.
Briefly, the algorithm finds the minimum mean cycle in an auxiliary graph whose nodes are 01-sequences of length $s$ that are valid shooting patterns for $\F$, and where a directed edge $x\to y$ connects sequences such that the last $s-1$ entries of $x$ coincide with the first $s-1$ entries of $y$.\\

\noindent\textbf{Proof of Lemma~\ref{lem:scaling}.}
Any shooting pattern $X$ for $\F$ can be
scaled by a factor of $d$ and used $d$ times (once for each remainder mod $d$) to give a shooting pattern $X'$ for $\F'=d\F$ with the same density.
(Formally, we set $x'_{kd+r}=x_k$ for any $k\in\Z$ and $0\le r\le d-1$.)
And vice versa, starting with a shooting pattern $X'$ for $\F'$,
for each $0\le r\le d-1$, consider the pattern $X^r$ defined by $x^r_i=x'_{id+r}$.
Then each $X^r$ hits $\F$ and the sparsest of them has density at most $\den(X)$.

\section{Proof of bounds for $m_k^n$ and 
$M_k^n$}\label{app:general-ub}

Here we state and prove our general bounds on the required
densities $m^n_k$ and $M^n_k$ for the simplest and toughest instances among families consisting of $n$ ships of size $k$ each.
%density $M^n_k$ for the toughest instances among families consisting of $n$ ships of size $k$ each.
\begin{theorem*}
Let $n\ge 1$ and $k\ge 2$ be integers.
Then $m_k^n=\frac{1}{k}$ and 
\[1-\frac{e}{\sqrt[k-1]{n}} \le M_k^n\le \min\left\{\frac{n}{n+1},\frac{1+\log(kn)}{k}\right\}.\]
\end{theorem*}

\begin{proof}
We start with the proof for $m_k^n=\frac{1}{k}$.
First, we prove that $\den(S)\geq \frac{1}{k}$ for any single $k$-ship $S$.
Let $s=\spn(S)$ be the span of $S$,
take $n$ large and consider a block $B_n$ of length $n+s-1$.
Then $B_n$ contains $n$ translates of $S$.
Since any shot in $B_n$ hits at most $k$ translates of $S$ in $B_n$, there must be at least $n/k$ shots in $B_n$.
Thus the density of shots within $B_n$ is at least 
$(n/k)/|B_n| \to_{n\to\infty} \frac1k$.

On the other hand, for any $n,k\in\N$ the family \[\F=\{[1,\dots,k-1,k], [1,\dots,k-1,2k],\dots,[1,\dots,k-1,nk]\}
\]
can be hit
by a shooting pattern $X=\{c\cdot k\mid c\in\Z\}$ that has density $\den(X)=1/k$.

Regarding $M_k^n$, we show the lower bound first.
The idea is to consider a block $B$ of a suitable length $m$ and define $\F$ as the set of all $ \binom{m-1}{k-1}$ ships of size $k$ which fit into this block and contain its first cell.
Then every $k$-tuple within $B$ is either a ship of $\F$ or a translate of a ship of $\F$.
Hence, in order to hit $\F$ we have to hit $B$ at least $m-(k-1)$ times, which implies $\den(\F)\ge \frac{m-(k-1)}{m}$.

Now we calculate a suitable $m$. Note that we must make sure that $n\ge \binom{m-1}{k-1}$ and $m\in\N$.
Using $\binom{m-1}{k-1}\le \frac{(m-1)^{k-1}}{(k-1)!}$ it suffices to have
\begin{equation*}
n \geq \frac{(m-1)^{k-1}}{(k-1)!}
\quad \Leftrightarrow \quad
m \leq 1+\sqrt[k-1]{n(k-1)!}.
\end{equation*}
Setting $m=\lfloor 1+ \sqrt[k-1]{n(k-1)!}\rfloor$ and using a standard bound $t!\ge (t/e)^t$ for $t=k-1$, we~get
\[m=\lfloor 1+ \sqrt[k-1]{n(k-1)!} \rfloor > \sqrt[k-1]{n(k-1)!}\ge \frac{k-1}e\sqrt[k-1]{n},
\]
thus
\[
\den(\F)\geq\frac{m-(k-1)}{m}=1-\frac{k-1}m \ge 1-\frac{e}{\sqrt[k-1]{n}}.\]

For the upper bound,
the expression $n/(n+1)$ is inherited from~\cref{thm:n2}.
The second expression comes from calculations that N.Alon \cite{alon-1990} did to compute an upper bound on minimal transversals of hypergraphs.

Namely, we take a period $s \geq \spn(F)$ and we mark one cell of the ship with an anchor.
Obviously, we can define a bijection between the translates of a ship~$S$ modulo $s$ and $\{0,1,\dots, s-1\}$ by identifying each translate with point where the anchor of~$S$ is positioned. 
Now, we construct a digraph~$G = (V,E)$ with~$s$ vertices and add a directed edge with color $c \in F$ from point~$a$ to point~$b$ if a mark at point~$a$ hits the translate of $c$ with anchor at point~$b$.

Note that a feasible shooting pattern corresponds to a directed dominating set in each monochromatic subgraph induced by~$V$ of~$G$.
Further, all vertices in~$G$ have exactly~$k-1$ incoming and outgoing edges of the same color.
Now, we choose a random set $X$ from $\{0,1,\dots, s-1\}$ and for each $c\in F$, we add the set $Y_c$ that consist of elements where an anchor of a translate of $c$ that is not hit by $X$ is located.
Completely analogous calculations as in the proof of Theorem 1.2.2 in \cite{probmeth} then yield the desired upper bound of $(1+\log(kn))/k)$.
\end{proof}

Note that the proof also works for higher dimensions.

\section{Notes on generalizations to higher dimensions}\label{app:d=2}

In this section, we show how to generalize Theorems~\ref{thm:formula22}, \ref{thm:n2}, and \ref{thm:32} to two dimensions. %(or to any dimension, for that matter).
Interestingly, it turns out that all three results are actually \emph {easier} in $\Z^2$; intuitively this is because for small ships, either the ships are all in a one-dimensional subspace, or they can be decomposed into smaller components for which the reasoning is simpler. 
We believe that for two-dimensional ships of size $4$, or for at least three two-dimensional ships of size $3$, the situation in 2D does become more complex.

First, we generalize Theorem \ref{thm:formula22}.

\begin{theorem}[Formula for two 2-ships in 2D]\label{thm:formula22-2d}
Let $\F=\{[0,u],[0,v]\}$, for $u, v \in \Z^2$.
\begin{enumerate}
\item If $u$ and $v$ are linearly independent, then $\den(\F)=1/2$.
\item If $u=aw$ and $v=bw$ for some $w \in \Z^2$ and $a$, $b$ coprime and odd, then $\den(\F)=1/2$.
\item Otherwise, $\den(\F)=\frac{a+b+1}{2(a+b)}$.
\end{enumerate}
\end{theorem}
\begin{proof}
A 2-ship in $\mathbb{Z}^2$ is of the form $[(0,0), (x,y)]$. So, we can think of it as the vector $(x,y)$. Now suppose we have two two ships, represented by two vectors $u$ and $v$.\\
If $u$ and $v$ are linearly independent, then the optimal density is 1/2. To see this, consider the subgrid of cells of the form $(au + bv)$ for integers $a$ and $b$. Any shooting pattern needs to hit half of the cells of this grid (either the cells where $a+b$ is odd or where $a+b$ is even). Since we can tile the plane with independent copies of this subgrid, we have $\den(\F)\le1/2$. The other inequality is trivial.\\
If $u$ and $v$ are not independent, then we are in the 1D case: let $w$ be the vector such that $u = aw$ and $v = bw$ for coprime integers $a$ and $b$, then we can tile the plane with copies of the space generated by $w$ and apply the existing 1D result (\cref{thm:formula22}) to each copy.
\end{proof}

Second, we generalize Theorem \ref{thm:n2}.
\begin{theorem}[Toughest families of 2-ships in 2D]\label{thm:n2-2d}
Let $\F=\{[0,u_1],\dots,[0,u_n]\}$ for $u_1,\dots,u_n\in\Z^2$. Then $\den(\F)\le n/(n+1)$.
\end{theorem}

\begin{proof}
First, we divide the plane into the 4 quadrants and  mark the cells along the axes by $1$'s such that the patterns in the quadrants are independent of each other. These areas have measure 0, so the density depends only on the density of the patterns within the individual quadrants. \\
Now we create a shooting pattern for the first quadrant with the same greedy algorithm as in the $d=1$ case. We anchor the ships at their leftmost cell (if there a two leftmost cells at the lowest leftmost cell) which will be our reference point such that they are contained in a $s\times t$ box. Then we move the anchored cells along the plane as follows: Start at the leftmost lowest unmarked cell $(a,b)$ and go up $(a,b+2t)$ cells. Now, go from $(a+1,b)$ to $(a+1,b+t)$. Next $(a,b+2t+1)$ to $(a,b+2t)$, then $(a+1, b+t+1)$ to $(a+1,b+2t)$ and $(a+2,b)$ to $(a+2,b+t)$ and so on. If the anchored parts are at $0$, we mark all cells where unanchored parts of the ships are located. By our moving pattern, we avoid marking any checked $0$'s.  %See Figure \ref{fig:greedy} for an illustration.
Since we add at most $n$ $1$’s at a $0$, the density of the pattern is at most $n/(n+1)$. Similarly, we can create patterns for all other quadrants.
\end{proof}

Third, we generalize Theorem \ref{thm:32}.
\begin{theorem}[Toughest 3-ship with its reflection in 2D]\label{thm:32-2d}
Given $u,v\in\Z^2$, let $S=[0,u,v]$ be a 3-ship in 2D,
let $\reverse{S}=[0,-u,-v]$ its reflection, and
set $\F=\{S,\reverse{S}\}$.
\begin{enumerate}
    \item If $u$ and $v$ are linearly independent, then $\den(\F)= 1/3$.
    \item Otherwise $\den(\F)\le \frac25$.
\end{enumerate}
\end{theorem}
\begin{proof}
If $u$ and $v$ are linearly independent then, similarly to~\cref{thm:formula22-2d}, consider the lattice $\Lambda=\{au+bv\mid a,b\in\Z\}$ generated by $u$ and $v$.
Within the lattice $\Lambda$, the ship $S$ is an L-shaped triomino, and $\reverse{S}$ is its rotation by $180^\circ$, thus, as in Lemma~\ref{lem:2l},
the subset $X=\{au+bv \mid a-b\equiv0\pmod 3\}$ is a shooting pattern for $\F$.
Since $\Z^2$ is a disjoint union of copies of $\Lambda$, using this shooting pattern for each copy of $\Lambda$ we get the desired $\den(\F)=\frac13$.

The second claim follows immediately from~\cref{thm:32}.
\end{proof}

\end{document}